\documentclass[twoside,11pt]{article}
\usepackage[preprint]{jmlr2e}
\usepackage[utf8]{inputenc}
\usepackage{amsmath}
\usepackage{amsfonts,color}
\usepackage{graphicx,subcaption}
\usepackage{bbm}
\usepackage{tikz}

\newtheorem{thm}{Theorem}

\newtheorem{prop}{Proposition}
\newtheorem{cor}{Corollary}
\newtheorem{lem}{Lemma}
\newtheorem{rem}{Remark}

\ShortHeadings{Minimax Bounds for Distributed Logistic Regression}{Barnes and \"Ozg\"ur}
\firstpageno{1}

\begin{document}

\title{Minimax Bounds for Distributed Logistic Regression}

\author{\name Leighton Pate Barnes \email lpb@stanford.edu \\
       \addr Department of Electrical Engineering\\
       Stanford University\\
       Stanford, CA 94305, USA
       \AND
       \name Ayfer \"Ozg\"ur \email aozgur@stanford.edu \\
       \addr Department of Electrical Engineering\\
       Stanford University\\
       Stanford, CA 94305, USA
       }

\editor{}


\maketitle

\begin{abstract}
We consider a distributed logistic regression problem where labeled data pairs $(X_i,Y_i)\in \mathbb{R}^d\times\{-1,1\}$ for $i=1,\ldots,n$ are distributed across multiple machines in a network and must be communicated to a centralized estimator using at most $k$ bits per labeled pair. We assume that the data $X_i$ come independently from some distribution $P_X$, and that the distribution of $Y_i$ conditioned on $X_i$ follows a logistic model with some parameter $\theta\in\mathbb{R}^d$. By using a Fisher information argument, we give minimax lower bounds for estimating $\theta$ under different assumptions on the tail of the distribution $P_X$. We consider both $\ell^2$ and logistic losses, and show that for the logistic loss our sub-Gaussian lower bound is order-optimal and cannot be improved.
\end{abstract}


\section{Introduction}

Modern datasets are often distributed and processed across multiple machines in a network which allows for a reduction in computation time by taking advantage of data parallelism. There are also an increasing number of applications where data is generated at remote nodes and learning a model from it requires exchanging information between the nodes over potentially slow and unreliable wireless links (e.g., \cite{federated0}). It is now well-understood that communication bandwidth can be an important bottleneck in the performance of such distributed machine learning systems. This has led to significant recent interest in techniques that demonstrate improvements in the communication efficiency of distributed learning (e.g., \cite{deepgrad,federated2}).

From a statistical learning  perspective, these observations motivate a fundamental problem: understanding the tradeoff between communication cost and the achievable performance for a given statistical task.  In a common, simplified setting, we consider i.i.d. samples $X_i$ for $i=1,\ldots,n$ which are distributed across different machines in a network and must be communicated across rate-limited links. For example, each sample could be communicated with $k$ bits to a centralized estimator. In the works \cite{duchi,braverman, garg}, the focus is on the Gaussian location model and its variants, where the samples $X_i$ are taken from a normal distribution $\mathcal{N}(\theta,\sigma^2 I_d)$ and the statistical task is to estimate the mean parameter $\theta$. In other works, the statistical task of interest is estimating a high-dimensional discrete distribution \citep{diakonikolas,yanjun2} or distributed property testing \citep{archayaetal,archayaetal2}. Of particular note are recent works by \cite{yanjun,allerton,isit} that prove general minimax lower bounds that apply to many distributed statistical estimation problems.

Logistic regression is a powerful and commonly used binary classification tool \citep{Hastie--Tibshirani--Friedman2009elements}. With it, we model labeled data points $(x,y)\in\mathbb{R}^d\times\{-1,1\}$ as being approximately separable by a hyperplane into two classes corresponding to the labels $y$, with a point having a greater likelihood of having a given label if it is further from the hyperplane on the corresponding side. If we have an estimate $\hat\theta$ of the hyperplane's normal vector, then given a new sample $x$ we can predict its label via the inner product $\langle \hat\theta,x \rangle$. In this paper, we aim to study the trade-off between the communication cost, i.e. the number of bits we are allowed to transmit about each labeled data pair $(X,Y)$, and the minimax risk of the underlying estimation problem in distributed logistic regression. There has been some recent work on characterizing the minimax risk for logistic regression under privacy constraints \citep{duchi2019lower}, but to the authors' knowledge this is the first work that considers this problem under communication constraints. We give minimax lower bounds that are general in that they apply to many distributions of $X$, whereas \cite{duchi2019lower} consider only $X\sim\text{unif}\{-1,1\}^d$, and we show that in some cases our bounds are order-optimal.

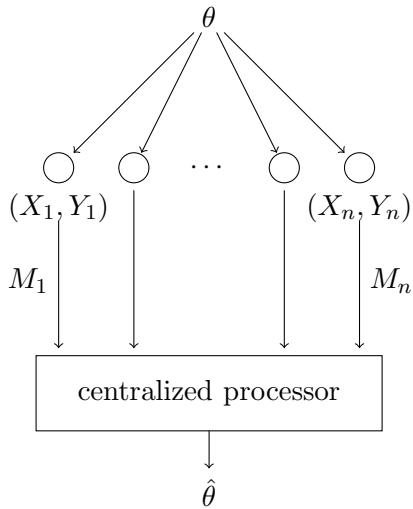
\begin{figure}
\centering{
\begin{tikzpicture}
\node at (2, 2) {$\theta$}; 
\draw [->] (1.8, 1.8) -- (0.2, 0.2); 
\draw [->] (1.9, 1.8) -- (1.1, 0.2); 
\draw [->] (2.1, 1.8) -- (2.9, 0.2); 
\draw [->] (2.2, 1.8) -- (3.8, 0.2); 
\draw (0,0) circle (0.2cm); \node [below] at (0,-0.2) {$(X_1,Y_1)$}; 
\draw (1,0) circle (0.2cm); 
\node at (2,0) {$\cdots$}; 
\draw (3,0) circle (0.2cm);
\draw (4,0) circle (0.2cm); \node [below] at (4,-0.2) {$(X_n,Y_n)$}; 
\draw [->] (0, -0.7) -- (0, -2.4); \draw [->] (1, -0.3) -- (1, -2.4); 
\draw [->] (3, -0.3) -- (3, -2.4); \draw [->] (4, -0.7) -- (4, -2.4);
\draw (-0.3,-3.5) rectangle (4.3,-2.5);  
\node at (2,-3) {centralized processor}; 
\draw [->] (2,-3.5) -- (2, -4); \node [below] at (2,-4) {$\hat{\theta}$}; 
\node [left] at (0, -1.5) {$M_1$}; \node [right] at (4,-1.5) {$M_n$}; 
\end{tikzpicture}
}
\caption{The distributed logistic regression problem, where the data are labeled pairs $(X_i,Y_i)$ that are encoded into $k$-bit messages $M_i$. The parameter $\theta$ is the parameter for the logistic model $P_{Y|X}$.}
\label{fig1b}
\end{figure}

\subsection{Preliminaries}
We consider a distributed logistic regression problem where labeled data pairs $(X_i,Y_i)$ for $i=1,\ldots,n$ are distributed across multiple machines in a network and must be communicated to a centralized estimator using at most $k$ bits per labeled pair. Suppose that the data $X_i\in\mathbb{R}^d$ are taken i.i.d. from some distribution $P_X$ with density $f(x)$ with respect to a dominating probability measure $\nu$. The labels $Y_i \in \{-1,1\}$ are then taken from a logistic model conditioned on $X_i$, i.e.
\begin{align*}
p_\theta(y_i|x_i) = \frac{1}{1+\exp(-y_i\langle\theta,x_i\rangle)}
\end{align*}
independently of each other, where $\theta$ is a $d$-dimensional parameter from a set of possible values $\Theta\subseteq\mathbb{R}^d$.

The goal in the logistic regression problem is to construct an estimate $\hat\theta$ of the true parameter $\theta$ that minimizes some loss function. This can also be thought of as estimating the hyperplane that best separates the data points with label one from those with label minus one. Each sample $(X_i,Y_i)$ is encoded into $k$ bits  by a (possibly stochastic) function $b_i:\mathbb{R}^d\times \{-1,1\} \to [2^k]$, and the estimator $\hat\theta$ will be a function of the encoded $M_i = b_i(X_i,Y_i)$ for $i=1,\ldots,n$. See Figure \ref{fig1b}. For simplicity, here we consider communication with no interactions between different nodes so that the messages $M_i$ are independent of each other, but our lower bounds can be extended to more general communication strategies that allow for arbitrary interaction between the samples such as blackboard protocols.

We will consider two different loss functions for the error incurred by our estimate $\hat\theta$. The first is squared $\ell^2$ loss where we are interested in characterizing the minimix risk
$$\inf_{(\hat\theta,\{b_i\})}\sup_{\theta\in\Theta}\mathbb{E}_\theta\|\hat\theta-\theta\|_2^2 \; .$$
We use $\mathbb{E}_\theta$ to denote taking the expectation over the $(X_i,Y_i)$ with respect to the joint density
$$p_\theta(x,y) = f(x)p_\theta(y|x) \; .$$
The infimum above is taken over all possible estimators $\hat\theta(M_1,\ldots,M_n)$ and all possible encoding strategies $b_i$ for $i=1,\ldots,n$. We are able to prove lower bounds for this $\ell^2$ risk by using the quantized/distributed Fisher information framework we first introduced in \cite{allerton}. We will also consider a generalized excess risk where we are interested in characterizing
$$\inf_{(\hat\theta,\{b_i\})}\sup_{\theta\in\Theta} \left(\mathbb{E}_\theta[R_\theta(\hat\theta)] - \min_{\theta_0} R_\theta(\theta_0)\right)$$
where
$$R_\theta(\hat\theta) = \mathbb{E}_\theta[\ell(\hat\theta;(X,Y))]$$
is the population risk associated with some loss $\ell$ that is convex in its first argument. In particular we will focus on the logistic loss
$$\ell(\theta;(x,y)) = \log(1+\exp(-y\langle\theta,x\rangle)$$
and show that in this case, a strong convexity argument from \cite{duchi2019lower} can be used to extend the Fisher information results to this new loss function. One reason this is noteworthy is that it is an example where Fisher information can be used to prove minimax lower bounds for a non-$\ell^2$ loss.

We use the super-exponential definition for sub-Gaussian and sub-exponential random variables and say that a random variable $X$ is $\sigma^2$-sub-Gaussian if $\mathbb{E}[\exp\frac{X^2}{\sigma^2}] \leq 2$, and a random variable $X$ is sub-exponential with parameter $\sigma$ if $\mathbb{E}[\exp\frac{|X|}{\sigma}] \leq 2$ \citep{versh}.

We say a random vector $X\in\mathbb{R}^d$ is $\sigma^2$-sub-Gaussian if the projection of $X$ onto any unit vector $u\in\mathbb{R}^d$ is a sub-Gaussian random variable with parameter $\sigma^2$. Similarly, we will say a random vector $X\in\mathbb{R}^d$ is sub-exponential with parameter $\sigma$ if the projection of $X$ onto any unit vector $u\in\mathbb{R}^d$ is a sub-exponential random variable with parameter $\sigma$.

\section{Minimax Results} \label{sec:results}

In this section we state our main minimax risk bounds for the communication-constrained logistic regression problem. For squared $\ell^2$ risk, we give three different lower bounds depending on the tail behavior of the random vector $X$, i.e. whether $X$ comes from a sub-Gaussian, sub-exponential, or finite second moment distribution. Stronger assumptions on the tail of $X$ lead to stronger (larger) lower bounds. In the sub-Gaussian case, we show that this lower bound also applies to a generalized excess risk with the logistic loss. We also show an achievable scheme (an encoding strategy along with an estimator $\hat\theta$) for $X\sim \text{unif}\{-1,1\}^d$ that gives an upper bound that matches the sub-Gaussian lower bound. This shows that for the logistic loss, the sub-Gaussian lower bound is order-optimal and cannot be improved.

\begin{thm}[sub-Gaussian lower bound] \label{thm:subg}
Suppose that $X\sim P_X$ is a $\sigma^2$-sub-Gaussian random vector. Then
$$\sup_{\theta\in\mathbb{R}^d}\mathbb{E}_\theta\|\hat\theta-\theta\|_2^2 \geq c\max\left\{\frac{d}{n\sigma^2},\frac{d^2}{kn\sigma^2}\right\}$$
for any estimator $\hat\theta$ and encoding strategies $\{b_i\}_{i=1,\ldots,n}$, and some absolute constant $c>0$ that is independent of $n,k,d,\sigma^2$.
\end{thm}
The proof method for Theorem \ref{thm:subg} and the subsequent $\ell^2$ lower bounds is sketched in Section \ref{sec:FI}. For the logistic loss, the following corollary is implied by a strong convexity argument detailed in Section \ref{sec:convexity}.
\begin{cor}[sub-Gaussian lower bound with logistic loss] \label{cor:subg}
Suppose that $X\sim P_X$ is a $\sigma^2$-sub-Gaussian random vector whose autocorrelation matrix has a minimum eigenvalue $\lambda_\text{min}(\mathbb{E}[XX^T]) \geq \delta > 0$. Then if $nk \geq d^4\sigma^4\log(d\sigma)$,
$$\sup_{\theta\in\mathbb{R}^d} \left(\mathbb{E}_\theta[R_\theta(\hat\theta)] - \min_{\theta_0} R_\theta(\theta_0)\right) \geq \delta c\max\left\{\frac{d}{n\sigma^2},\frac{d^2}{kn\sigma^2}\right\}$$
for any estimator $\hat\theta$ and encoding strategies $\{b_i\}_{i=1,\ldots,n}$, and some absolute constant $c>0$ that is independent of $n,k,d,\sigma^2,\delta$.
\end{cor}

The minimum eigenvalue assumption is a mild technical condition that is satisfied by many distributions of interest. For example, if $P_X$ is a product distribution with mean zero then this condition is trivially satisfied with $\delta$ being the minimum variance of the different components. The condition on $n$ is also a technical assumption that will be needed to ignore the second term in the denominator of \eqref{eq:vt2} below.

If we relax the sub-Gaussian assumption from Theorem \ref{thm:subg} to a sub-exponential or finite second-moment assumption, then instead of the linear dependence on $k$ in the denominator we get a quadratic or exponential dependence on $k$, respectively. This can be seen in the following two theorems.

\begin{thm}[sub-exponential lower bound] \label{thm:sube}
Suppose that $X\sim P_X$ is a sub-exponential random vector with parameter $\sigma$. Then
$$\sup_{\theta\in\mathbb{R}^d}\mathbb{E}_\theta\|\hat\theta-\theta\|_2^2 \geq c\max\left\{\frac{d}{n\sigma^2},\frac{d^2}{k^2n\sigma^2}\right\}$$
for any estimator $\hat\theta$ and encoding strategies $\{b_i\}_{i=1,\ldots,n}$, and some absolute constant $c>0$ that is independent of $n,k,d,\sigma^2$.
\end{thm}

\begin{thm}[finite second-moment lower bound] \label{thm:fsm}
Suppose that $X \sim P_X$ in $\mathbb{R}^d$ is such that $\mathbb{E}[\langle u,X\rangle^2] \leq I_0$ for any unit vector $u\in\mathbb{R}^d$. Then
$$\sup_{\theta\in\mathbb{R}^d}\mathbb{E}_\theta\|\hat\theta-\theta\|_2^2 \geq c\max\left\{\frac{d}{n\sigma^2},\frac{d^2}{2^knI_0}\right\}$$
for any estimator $\hat\theta$ and encoding strategies $\{b_i\}_{i=1,\ldots,n}$, and some absolute constant $c>0$ that is independent of $n,k,d,\sigma^2$.
\end{thm}

\subsection{An Upper Bound}
In order to demonstrate that the sub-Gaussian lower bound from Corollary \ref{cor:subg} is order-optimal for at least some sub-Gaussian distributions, we will consider the special case $X\sim\text{unif}\{-1,1\}^d$. In this case, if we consider the class-conditional distributions
\begin{align*}
p(x_j|y) & = \frac{e^{y \theta_jx_j/2}}{e^{\theta_j/2}+e^{-\theta_j/2}} = \begin{cases} \frac{e^{\theta_j/2}}{e^{\theta_j/2}+e^{-\theta_j/2}} & \text{ , if } x_j=y \\ \frac{e^{-\theta_j/2}}{e^{\theta_j/2}+e^{-\theta_j/2}} & \text{ , if } x_j\neq y\end{cases}
\end{align*}
for $j=1,\ldots,d$ as independent and identically distributed conditional on $Y$, and $Y\sim\text{unif}\{\pm1\}$, then $Y|X$ follows 
$$p_\theta(y|x) = \frac{\exp(y\langle\theta,x\rangle/2)}{\exp(y\langle\theta,x\rangle/2)+\exp(-y\langle\theta,x\rangle/2)} \; .$$
This is a logistic model with parameter $\theta$, so this setup matches our original joint distribution for $X,Y$. This means that if we only consider the first $k-1$ components of $X$, then they will also follow a logistic model with parameter equal to the first $k-1$ components of $\theta$. Using this observation, consider the following encoding strategy when $2\leq k \leq d$:
\begin{itemize}
\item[(i)] Partition the $d$ components into $m=\frac{d}{k-1}$ groups of $k-1$ components each. We assume $m$ is an integer for simplicity.
\item[(ii)] Assign each sample $X_i$ to one group such that each group is assigned $\frac{n}{m} = \frac{n(k-1)}{d}$ samples.
\item[(iii)] Encode each sample $X_i$ into a $k$-bit message $M_i = b_i(X_i)$ by picking out the $k-1$ components associated with its group and also appending the label $Y_i$.
\end{itemize}
With this strategy, the centralized estimator $\hat\theta$ can then solve $m$ different logistic regression problems each of dimension $k-1$ with an effective sample size of $\frac{n(k-1)}{d}$ for each problem. Using a strategy such as stochastic gradient descent, each problem can be solved with excess logistic risk that is asymptotically order $\frac{d(k-1)}{n(k-1)} = \frac{d}{n}$ (see \cite{bhowmick} Corollary 3.1, or \cite{polyak} Theorem 2), and thus in total
$$\mathbb{E}_\theta[R_\theta(\hat\theta)] - \min_{\theta_0} R_\theta(\theta_0) \leq C\frac{d^2}{kn}$$
for large enough $n$ and an absolute constant $C$.


\section{Fisher Information} \label{sec:FI}

The strategy for proving the minimax lower bounds from Section \ref{sec:results} will be to use the van Trees inequality along with a characterization of the Fisher information from the received messages $M_1,\ldots,M_n$. The van Trees inequality is a Bayesian version of the Cram\'er-Rao lower bound, which in the version we will use from \cite{gill}, bounds the average $\ell^2$ risk by
\begin{align} \label{eq:vt}
\int_\Theta \mathbb{E}_\theta\|\hat\theta-\theta\|_2^2\mu(\theta)d\theta \geq \frac{d^2}{\text{Tr}(I_{M_1,\ldots,M_n}(\theta)) + J(\mu)}
\end{align}
where $\mu(\theta) = \prod_{j=1}^d\mu_j(\theta_j)$ is a prior for the parameter $\theta$, $I_{M_1,\ldots,M_n}(\theta)$ is the Fisher information matrix for estimating $\theta$ from the messages $M_1,\ldots,M_n$, and $J(\mu)$ is the Fisher information associated with the prior $\mu$.

Recall that in the context of Fisher information, the score function associated with the statistical model $p_\theta(x,y)$ is
\begin{align*}
S_\theta(x,y) & = \nabla_\theta\log p_\theta(x,y) \\
& = \left( \frac{\partial}{\partial\theta_1}\log p_\theta(x,y),\ldots,\frac{\partial}{\partial\theta_d}\log p_\theta(x,y)\right) \; .
\end{align*}
The Fisher information matrix for estimating $\theta$ from an encoded sample $M$ has entries
$$[I_M(\theta)]_{i,j} = \mathbb{E}_\theta\left[\frac{\partial}{\partial\theta_i}\log p_\theta(m)\frac{\partial}{\partial\theta_j}\log p_\theta(m)\right]$$
and the Fisher information from the prior is
$$J(\mu) = \sum_{i=1}^d \int \frac{\mu_i'(\theta_i)^2}{\mu_i(\theta_i)}d\theta_i  \; .$$
Assuming that $\Theta = [-B,B]^d$, the prior $\mu$ can be chosen to minimize $J(\mu)$ \citep{Tsybakov2008,borovkov}. This observation along with the independence of the $M_i$, and upper bounding the average risk by the maximum risk, leads to
\begin{align} \label{eq:vt2}
\sup_{\theta\in\Theta} \mathbb{E}_\theta\|\hat\theta-\theta\|_2^2 \geq \frac{d^2}{n\mathsf{Tr}(I_{M}(\theta)) + \frac{d\pi^2}{B^2}} \; .
\end{align}

Proving Theorems \ref{thm:subg}-\ref{thm:fsm} therefore boils down to understanding and upper bounding $\mathsf{Tr}(I_M(\theta))$, i.e. the Fisher information from a single sample. In light of this, we have the following lemmas:
\begin{lem} \label{lem1}
If for any $\theta\in\Theta$, $S_\theta(X,Y)$ is a $\sigma^2$-sub-Gaussian random vector, then
$$\mathsf{Tr}(I_M(\theta)) \leq \min\{ \mathsf{Tr}(I_{X,Y}(\theta)), \; Ck\sigma^2\}, $$
for an absolute constant $C$.
\end{lem}
\begin{lem} \label{lem2}
If for any $\theta\in\Theta$, $S_\theta(X,Y)$ is a sub-exponential random vector with parameter $\sigma$, then
$$\mathsf{Tr}(I_M(\theta)) \leq \min\{ \mathsf{Tr}(I_{X,Y}(\theta)), \; Ck^2\sigma^2\}, $$
for an absolute constant $C$.
\end{lem}
\begin{lem} \label{lem3}
If for any $\theta\in\Theta$ and any unit vector $u\in\mathbb{R}^d$, $$\mathsf{Var}(\langle u, S_\theta(X,Y)\rangle) \leq I_0 \; ,$$  then
$$\mathsf{Tr}(I_M(\theta)) \leq \min\{ \mathsf{Tr}(I_{X,Y}(\theta)), \; 2^kI_0\} \; .$$
\end{lem}
These three lemmas show how the Fisher information from a single sample $M$ can scale with the number of bits $k$. Depending on the tail behavior of the score function random vector $S_\theta(X,Y)$, the scaling can be at most linear, quadratic, or exponential in $k$. The proofs of these lemmas are included in the last section of the paper. Luckily, for a logistic model, the tail behavior of the score function random vector $S_\theta(X,Y)$ cannot be any worse than the corresponding tail behavior for the data $X$. This is shown in the following propositions. Theorems \ref{thm:subg}-\ref{thm:fsm} then follow by applying Lemmas \ref{lem1}-\ref{lem3} to equation \eqref{eq:vt2}, respectively, and taking $B\to\infty$.

\begin{prop}
Suppose that $X \sim P_X$ is a $\sigma^2$-sub-Gaussian random vector in $\mathbb{R}^d$. Then the score function $S_\theta(X,Y)$ is also a $\sigma^2$-sub-Gaussian random vector in $\mathbb{R}^d$.
\end{prop}
\begin{proof}
The joint distribution of $X$ and $Y$ can be written as the product
\begin{align*}
p_\theta(x,y) & = f(x)\frac{1}{1+\exp(-y\langle\theta,x\rangle)}
\end{align*}
so that the score for component $\theta_i$ is
\begin{align*}
S_{\theta_i}(x,y)  = \frac{\partial}{\partial\theta_i} \log p_\theta(x,y)  = \frac{\exp(-y\langle\theta,x\rangle)}{1+\exp(-y\langle\theta,x\rangle)}yx_i \; .
\end{align*}
Projecting the score function vector $S_\theta(x,y)$ onto any unit vector $u$ then gives
\begin{align} \label{eq:prop1}
\langle u,S_\theta(x,y) \rangle & = \sum_{i=1}^d \frac{\exp(-y\langle\theta,x\rangle)}{1+\exp(-y\langle\theta,x\rangle)}yx_iu_i \nonumber \\
& = y\left(\frac{\exp(-y\langle\theta,x\rangle)}{1+\exp(-y\langle\theta,x\rangle)}\right)\langle u,x\rangle \; .
\end{align}
Note that the prefactor in \eqref{eq:prop1} has magnitude less than or equal to one, i.e.
$$\left|y\left(\frac{\exp(-y\langle\theta,x\rangle)}{1+\exp(-y\langle\theta,x\rangle)}\right)\right| \leq 1 \; ,	$$
so that
\begin{align} \label{eq:prop1_2}
\mathbb{P}\left(|\langle u,S_\theta(X,Y)\rangle|\geq t\right) & \leq \mathbb{P}\left(|\langle u,X\rangle|\geq t\right) \nonumber \\
& \leq 2\exp\left(\frac{-t^2}{K\sigma^2}\right)
\end{align}
where $K$ is an absolute constant and \eqref{eq:prop1_2} follows from the sub-Gaussianity of $X$.
\end{proof}
Using nearly identical proofs, we get similar results for sub-exponential and finite second-moment random vectors $X$:
\begin{prop}
Suppose that $X \sim P_X$ is a sub-exponential random vector in $\mathbb{R}^d$ with parameter $\sigma$. Then the score function $S_\theta(X,Y)$ is also a sub-exponential random vector in $\mathbb{R}^d$ with parameter $\sigma$.
\end{prop}
\begin{prop}
Suppose that $X \sim P_X$ in $\mathbb{R}^d$ is such that $\mathbb{E}[\langle u,X\rangle^2] \leq I_0$ for any unit vector $u\in\mathbb{R}^d$. Then the score function $S_\theta(X,Y)$ also satisfies $\mathbb{E}[\langle u,S_\theta(X,Y)\rangle^2] = \text{var}(\langle u,S_\theta(X,Y)\rangle) \leq I_0$ for any unit vector $u\in\mathbb{R}^d$.
\end{prop}

\begin{rem}
In the sub-Gaussian case, when the Fisher information is at most linear with $k$, we can easily extend our lower bounds to the case where an average of $k$ bits per sample is sent to the centralized estimator, i.e. $nk$ bits in total, but we do not require a strict $k$-bit budget for each sample. This will not change the resulting upper bound on $I_{M_1,\ldots,M_n}(\theta)$.
\end{rem}
\begin{rem}
In another variation, we could consider a ``batch'' communication budget where there are $m$ machines each with a batch of $l$ independent samples, and each machine can communicate $k$ bits to the centralized estimator. In this case the score functions combine linearly, i.e. $S_\theta(X_1,X_2,Y_1,Y_2) = S_\theta(X_1,Y_1) + S_\theta(X_2,Y_2)$, and we will arrive at the same lower bounds (in terms of $n,d,k,\sigma$) if we set $n=lm$. This is interesting in that it accounts for potential computations that a machine could perform jointly on multiple samples.
\end{rem}

\section{Logistic Loss (Proof of Corollary \ref{cor:subg})} \label{sec:convexity}
In this section we will use a convexity argument similar to that from \cite{duchi2019lower} in order to show how Corollary \ref{cor:subg} follows in a similar way to Theorem \ref{thm:subg}. The goal will be to lower bound the excess risk associated with the logistic loss by a quadratic in a small neighborhood of $\hat\theta$ values, so that we can use the same Fisher information argument from Section \ref{sec:FI}.

We will restrict our attention to a subset of $\theta$ values and only consider $\theta \in \left[-\frac{r}{4\sqrt{d}},\frac{r}{4\sqrt{d}}\right]^d$ for an $r>0$ to be specified later. Since we take a supremum over $\theta$ values this cannot increase the minimax risk. If $R_\theta(\hat\theta)$ were $\lambda$-strongly convex as a function of $\hat\theta$ for $\|\hat\theta\|_2 \leq r$, then
\begin{align*}
R_\theta(\hat\theta) - R_\theta(\theta_0) \geq \nabla R_\theta(\theta_0)(\hat\theta-\theta_0) + \frac{\lambda}{2}\|\hat\theta -\theta_0\|_2^2
\end{align*}
for any $\|\hat\theta\|_2\leq r$ and $\|\theta_0\|_2\leq r$. Picking $\theta_0 = \theta$ note that $\nabla R_\theta(\theta) = 0$ and we would get
\begin{align*}
R_\theta(\hat\theta) - \inf_{\theta_0}R_\theta(\theta_0) \geq \frac{\lambda}{2}\|\hat\theta -\theta\|_2^2
\end{align*}
for any $\|\hat\theta\|_2\leq r$. Furthermore, since $R_\theta(\hat\theta)$ is convex for all $\hat\theta$, we would have
\begin{align} \label{eq:cvx}
R_\theta(\hat\theta) - \inf_{\theta_0}R_\theta(\theta_0) \geq \begin{cases}\frac{\lambda}{2}\|\hat\theta -\theta\|_2^2 & \; \text{if} \; \|\hat\theta-\theta\|_2 \leq r/2 \\ \lambda \frac{r}{2} \|\hat\theta-\theta\|_2 - \frac{\lambda r^2}{8} & \; \text{if} \; \|\hat\theta - \theta\|_2 \geq r/2\end{cases}
\end{align}
for all $\hat\theta$. Because $\|\theta\|_2\leq r/4$, the right-hand side of \eqref{eq:cvx} cannot increase when projecting $\hat\theta$ onto the ball $\|\hat\theta\|_2\leq r/4$, so we will without loss of generality assume that $\|\hat\theta\|_2\leq r/4$ and thus $\|\hat\theta-\theta\|_2 \leq r/2$ and
$$R_\theta(\hat\theta) - \inf_{\theta_0}R_\theta(\theta_0) \geq \frac{\lambda}{2}\|\hat\theta -\theta\|_2^2 \; .$$
Our conclusion is that if $R_\theta(\hat\theta)$ is $\lambda$-strongly convex for $\|\hat\theta\|_2\leq r$, then we can lower bound the excess logistic risk by a quadratic as desired. We therefore need to show that the minimum eigenvalue of the Hessian $\nabla^2 R_{\theta}(\hat\theta)$ is bounded away from zero for $\|\hat\theta\|_2\leq r$.

Using Lebesgue's dominated convergence theorem,

\begin{align*}
\nabla^2 R_{\theta}(\hat\theta) & = \mathbb{E}_{\theta}[\nabla^2 \log(1+\exp(-Y\langle\hat\theta,X\rangle)] \\
& = \mathbb{E}_{\theta}\left[\frac{1}{1+\exp(-Y\langle\hat\theta,X\rangle)}\left(1-\frac{1}{1+\exp(-Y\langle\hat\theta,X\rangle)}\right)XX^T\right] \; .
\end{align*}

Since we are assuming that $\langle u,X\rangle$ is $\sigma^2$-sub-Gaussian for any unit vector $u\in\mathbb{R}^d$, we have in particular that $\langle \frac{\hat\theta}{\|\hat\theta\|},X\rangle$ is $\sigma^2$-sub-Gaussian. Therefore $\langle \hat\theta,X\rangle$ is $\|\hat\theta\|^2\sigma^2$-sub-Gaussian. By sub-Gaussian concentration, if $\|\hat\theta\|\leq r$ then
\begin{align*}
|\langle \hat\theta,X\rangle| \leq \sqrt{2r^2\sigma^2\log\frac{2}{\alpha}}
\end{align*}
with probability at least $1-\alpha$. Let $t=\sqrt{2r^2\sigma^2\log\frac{2}{\alpha}}$. For $0<\epsilon<1$, let $r$ be small enough such that $\frac{e^{t}}{(1+e^t)^2} \geq \frac{1-\epsilon}{4}$,
and therefore
\begin{align}
p_{\hat\theta}(Y|X)(1-p_{\hat\theta}(Y|X)) & =\frac{1}{1+\exp(-Y\langle\hat\theta,X\rangle)}\left(1-\frac{1}{1+\exp(-Y\langle\hat\theta,X\rangle)}\right) \nonumber\\
& \geq  \frac{1-\epsilon}{4} \label{eq:conc_ineq}
\end{align}
with probability at least $1-\alpha$. Then
\begin{align*}
\lambda_\text{min}\left(\mathbb{E}_\theta\left[p_{\hat\theta}(Y|X)(1-p_{\hat\theta}(Y|X))XX^T\right]\right) & \geq \lambda_\text{min}\left(\frac{1-\frac{\epsilon}{2}}{4}\mathbb{E}[XX^T]\right) + \lambda_\text{min}(A)
\end{align*}
where the matrix $A$ is the difference
\begin{align*}
A & = \mathbb{E}_\theta\left[\left(p_{\hat\theta}(Y|X)(1-p_{\hat\theta}(Y|X)) - \frac{1-\epsilon}{4}\right)XX^T\right] \; .
\end{align*}
The row-sums of the magnitudes of the entries of the $A$ matrix can be bounded (as we will show below) by
\begin{align} \label{eq:A_entry}
\sum_{j=1}^d|A_{ij}| \leq d\epsilon \sigma^2 + d\alpha\sigma^2\log\frac{2}{\alpha} \; .
\end{align}
Therefore $\lambda_\text{min}(A) \geq - d\epsilon\sigma^2 - d\alpha\sigma^2\log\frac{2}{\alpha}$, and
\begin{align*}
\lambda_\text{min}\left(\nabla^2 R_{\theta}(\hat\theta)\right) & \geq \frac{1-\epsilon}{4}\lambda_\text{min}(\mathbb{E}[XX^T]) - d\epsilon\sigma^2 - d\alpha\sigma^2\log\frac{2}{\alpha}.
\end{align*}
Finally, if we assume that $\lambda_\text{min}(\mathbb{E}[XX^T]) \geq \delta > 0$ and set $\epsilon,\alpha,r$ small enough then $\lambda_\text{min}\left(\nabla^2 R_{\theta}(\hat\theta)\right) \geq c_0\delta$ as needed. In particular we could set $\epsilon= \Theta(\sigma^{-2}d^{-1})$, $\alpha = \Theta(\sigma^{-3}d^{-3/2})$ and $r = \Theta(d^{-1}\sigma^{-3}[\log(d\sigma)]^{-1/2}) \; .$
In the lower bound proofs in Section \ref{sec:FI} we can therefore set $B= \Theta\left(d^{-3/2}\sigma^{-3}[\log(d\sigma)]^{-1/2}\right)$ and the corollary follows.

\subsection{Proof of Equation \eqref{eq:A_entry}}
Starting from the left-hand side of \eqref{eq:A_entry},
\begin{align*}
\sum_{j=1}^d|A_{ij}|  & = \sum_{j=1}^d \left| \mathbb{E}_\theta\left[\left(p_{\hat\theta}(Y|X)(1-p_{\hat\theta}(Y|X)) - \frac{1-\epsilon}{4}\right)X_iX_j\right] \right|\\
& \leq  \sum_{j=1}^d  \mathbb{E}_\theta\left[ \left|p_{\hat\theta}(Y|X)(1-p_{\hat\theta}(Y|X)) - \frac{1-\epsilon}{4}\right| |X_i|\left|X_j\right|\right] \\
& = \sum_{j=1}^d\int_{\mathbb{R}^d} \left|p_{\hat\theta}(1|x)(1-p_{\hat\theta}(1|x)) - \frac{1-\epsilon}{4}\right|  |x_i| |x_j| \, f(x)d\nu(x) \; .
\end{align*}
We now split the integral into two parts; one integral over the set
$$E = \left\{x\in\mathbb{R}^d \, : \, p_{\hat\theta}(1|x)(1-p_{\hat\theta}(1|x)) \geq \frac{1-\epsilon}{4}\right\}$$
and one over it's complement $E^C = \mathbb{R}^d\setminus E$. We have
\begin{align}
\sum_{j=1}^d|A_{ij}|  = & \sum_{j=1}^d \int_E \left|p_{\hat\theta}(1|x)(1-p_{\hat\theta}(1|x)) - \frac{1-\epsilon}{4}\right| |x_i||x_j| \, f(x)d\nu(x) \label{eq:term1}
\\
& + \sum_{j=1}^d \int_{E^C} \left|p_{\hat\theta}(1|x)(1-p_{\hat\theta}(1|x)) - \frac{1-\epsilon}{4}\right| |x_i||x_j| \, f(x)d\nu(x) \label{eq:term2}
\end{align}
The term \eqref{eq:term1} is upper bounded by $d\epsilon\sigma^2$ by the definition of $E$ and a bound on the second moment due to sub-Gaussianity. For \eqref{eq:term2} we note that
$$ \int_{E^C} |x_i||x_j| \, f(x)d\nu(x) \leq \max\left\{\int_{E^C} x_i^2\, f(x)d\nu(x),\int_{E^C} x_j^2 \, f(x)d\nu(x)\right\}$$
and each term on the right-hand side can be bounded by the super-exponential property of sub-Gaussian distributions \citep{versh}:
\begin{align*}
2 \geq \mathbb{E}\left[\exp\frac{X_i^2}{\sigma^2}\right] & \geq \mathbb{E}\left[1_{E^C}(X)\exp\frac{X_i^2}{\sigma^2}\right] =  \mathbb{E}[1_{E^C}(X)]\mathbb{E}\left[\exp\frac{X_i^2}{\sigma^2}\bigg|E^C\right] \\
& \geq \mathbb{E}[1_{E^C}(X)]\exp\left(\mathbb{E}\left[\frac{X_i^2}{\sigma^2}\bigg| E^C\right]\right) \\
& = \mathbb{E}[1_{E^C}(X)]\exp\left(\mathbb{E}\left[\frac{X_i^2}{\sigma^2}1_{E^C}(X)\right]\bigg/ \mathbb{E}[1_{E^C}(X)]\right)
\end{align*}
and thus
$$\mathbb{E}\left[X_i^2 1_{E^C}(X)\right] \leq \sigma^2\mathbb{E}[1_{E^C}(X)]\log\frac{2}{\mathbb{E}[1_{E^C}(X)]} \leq \sigma^2 \alpha\log\frac{2}{\alpha} \; .$$

\section{Proofs of Lemmas \ref{lem1}-\ref{lem3}}

In this section we focus on proving Lemmas \ref{lem1}-\ref{lem3}. The first term in the max term of each lemma is a simple consequence of the data processing inequality for Fisher information \citep{zamir}, so we will focus on the other terms. Using Lemma 2 from \cite{allerton}, we have
\begin{align}\label{eq:geom}
\mathsf{Tr}(I_M(\theta)) = \sum_{i=1}^d [I_M(\theta)]_{i,i} =\sum_m p_\theta(m)\|\mathbb{E}_\theta[S_\theta(X,Y)|m]\|_2^2 \; .
\end{align}

\subsection{Proof of Lemmas \ref{lem1} -\ref{lem2}}
Let
$$u = \frac{1}{\| \mathbb{E}[S_\theta(X,Y)|m] \|_2} \mathbb{E}_\theta[S_\theta(X,Y)|m]$$
and note that
$$\| \mathbb{E}_\theta[S_\theta(X,Y)|m]\|_2 = \frac{1}{p_\theta(m)}\mathbb{E}_\theta[\langle u_,S_\theta(X,Y)\rangle q_m(X,Y)]$$
where $q_m(x,y) = p(m|x,y)$ is the probability that a sample $(x,y)$ is encoded into message $m$ via the possibly stochastic quantization map.

In the sub-Gaussian case (Lemma \ref{lem1}) let $p=2$, and in the sub-exponential case (Lemma \ref{lem2}) let $p=1$. The tail bound on $\langle u,S_\theta(X,Y)\rangle$ and the convexity of $x\mapsto \exp(|x|^p)$ gives
\begin{align*}
2 & \geq \mathbb{E}_\theta[\exp(|\langle u,S_\theta(X,Y) \rangle|^p/\sigma^{p}] \\
& \geq \mathbb{E}_\theta[q_m(X,Y)\exp(|\langle u,S_\theta(X,Y) \rangle|^p/\sigma^{p}] \\
& = p_\theta(m)\mathbb{E}_\theta[\exp(|\langle u,S_\theta(X,Y) \rangle|^p/\sigma^{p}|m] \\
& \geq p_\theta(m)\exp\left( |\mathbb{E}_\theta[\langle u,S_\theta(X,Y) \rangle | m]|^p / \sigma^p\right)
\end{align*}
so that
$$\mathbb{E}_\theta[\langle u,S_\theta(X,Y) \rangle|m] \leq \sigma\left(\log\left(\frac{2}{p_\theta(m)}\right)\right)^\frac{1}{p} \; .$$
Therefore,
\begin{equation} \label{eq:tail_bound}
\|\mathbb{E}_\theta[S_\theta(X,Y)|m]\|_2\leq \sigma\left(\log\left(\frac{2}{p_\theta(m)}\right)\right)^\frac{1}{p} \; .
\end{equation}
Combining \eqref{eq:geom} and \eqref{eq:tail_bound}, 
$$\mathsf{Tr}(I_M(\theta)) \leq \sigma^2 \sum_m p_\theta(m) \left(\log\left(\frac{2}{p_\theta(m)}\right)\right)^\frac{2}{p} \; .$$
To bound this expression, let $\phi$ be the upper concave envelope of $x\mapsto x\left(\log\frac{2}{x}\right)^\frac{2}{p}$ on $[0,1]$. We have
\begin{align}
\mathsf{Tr}(I_M(\theta)) & \leq \sigma^22^k\sum_m \frac{1}{2^k}\phi(p_\theta(m)) \nonumber \\
& \leq \sigma^2 2^k \phi\left(\sum_m \frac{1}{2^k}p_\theta(m)\right) \nonumber \\
& = \sigma^2 2^k \phi\left(\frac{1}{2^k}\right) \; .
\end{align}
It can be easily checked that $\phi(x) = x\left(\log\frac{2}{x}\right)^\frac{2}{p}$ for $0<x\leq1/2$, and therefore
\begin{align*}
\mathsf{Tr}(I_M(\theta)) & \leq \sigma^2\left(\log 2^{k+1}\right)^\frac{2}{p} \\
& \leq \sigma^2(k+1)^\frac{2}{p} \\
& \leq 4\sigma^2k^\frac{2}{p} \; .
\end{align*}

\subsection{Proof of Lemma \ref{lem3}}
Using the Cauchy-Schwarz inequality,
\begin{align*}
p_\theta(m)\| \mathbb{E}_\theta[S_\theta(X)|m]\|_2^2 & = \frac{1}{p_\theta(m)}\mathbb{E}_\theta[\langle u,S_\theta(X,Y)\rangle q_m(X,Y)]^2 \\
& \leq \frac{1}{p_\theta(m)}\mathbb{E}_\theta[\langle u,S_\theta(X,Y)\rangle^2]\mathbb{E}_\theta[q_m(X,Y)^2] \\ 
& \leq \frac{1}{p_\theta(m)}\mathbb{E}_\theta[\langle u,S_\theta(X,Y)\rangle^2]\mathbb{E}_\theta[q_m(X,Y)] \\ 
& = \mathbb{E}_\theta[\langle u,S_\theta(X,Y)\rangle^2] \; .
\end{align*}
So if $\mathsf{Var}\langle u,S_\theta(X,Y) \rangle \leq I_0$, then because score functions have zero mean,
$$p_\theta(m)\| \mathbb{E}_\theta[S_\theta(X,Y)|m]\|_2^2 \leq I_0 \; .$$
Therefore by \eqref{eq:geom},
$$\mathsf{Tr}(I_M(\theta)) \leq 2^kI_0 \; .$$

\bibliographystyle{IEEEtran}
\bibliography{di.bib}

\end{document}